\documentclass[11pt]{article}

\usepackage{fullpage}
\usepackage{graphics,graphicx,varwidth}
\usepackage[dvips]{epsfig}
\usepackage{eepic}
\usepackage{amsmath,amssymb,amsfonts}
\usepackage{color}
\usepackage{amsmath,amssymb,amsfonts}
\usepackage[lined]{algorithm2e}
\usepackage{lmodern}
\usepackage{mathtools}
\usepackage{microtype}
\usepackage{ifpdf}
\usepackage{enumitem}
\usepackage[lined]{algorithm2e}
\usepackage{multirow}
\usepackage{microtype,mathtools,mathrsfs}
\usepackage{times}
\usepackage{amsmath}
\usepackage{amsfonts}
\usepackage{amssymb}
\usepackage{amsthm}
\usepackage[english]{babel}
\usepackage[latin1]{inputenc}
\usepackage{cite}
\newtheorem{theorem}{Theorem}[section]

\newtheorem{claim}[theorem]{Claim}
\newtheorem{lemma}[theorem]{Lemma}

\newtheorem{corollary}[theorem]{Corollary}

\usepackage[small,compact]{titlesec}
\usepackage[small,it]{caption}

\def\cF{{\cal F}}

\def\cP{{\cal P}}

\def\cR{{\cal R}}

\newcommand{\var}{\mbox{\rm var}}
\newcommand{\mean}{\mbox{\rm mean}}

\newcommand{\Conf}{\mbox{\tt \textcolor[gray]{0.3}{{\bf ALG}}}}
\newcommand{\CONF}{\mbox{\tt \textcolor[gray]{0.3}{{\bf ALG}}}}
\newcommand{\conf}{\mbox{\tt \textcolor[gray]{0.3}{{\bf ALG}}}}
\newcommand{\OPT}{\mbox{\tt \textcolor[gray]{0.3}{{\bf Opt}}}}
\newcommand{\opt}{\mbox{\tt \textcolor[gray]{0.3}{{\bf Opt}}}}

\date{}
\begin{document}

\title{Clock Synchronization and Distributed Estimation\\ in Highly Dynamic Networks: \\ An Information Theoretic Approach}




\author{
 Ofer Feinerman
\thanks{The Shlomo and Michla Tomarin Career Development Chair, The Weizmann Institute of Science, Rehovot, Israel. E-mail: {\tt  ofer.feinerman@weizmann.ac.il}.
Supported by the Clore Foundation, the Israel Science
Foundation (FIRST grant no. 1694/10) and the Minerva Foundation.}
 \and
 Amos Korman
\thanks{CNRS and University Paris Diderot, Paris, 75013, France.  E-mail: {\tt amos.korman@liafa.univ-paris-diderot.fr}. Supported in part by the ANR project DISPLEXITY. This work has received funding from the European Research Council (ERC) under the European Union's Horizon 2020 research and innovation programme (grant 
agreement No 648032).}
 }

\begin{titlepage}
\maketitle

\date{}

\def\thefootnote{\fnsymbol{footnote}}

\pagenumbering{arabic}
\thispagestyle{empty}

\begin{abstract}\parskip0.08cm

We consider the {\em External Clock Synchronization} problem in dynamic sensor networks. Initially, sensors obtain inaccurate estimations of an external time reference and subsequently collaborate in order to synchronize their internal clocks with the external time. 
For simplicity, we adopt the {\em drift-free} assumption, where internal clocks are assumed to tick at the same pace.  Hence, the problem is reduced to an estimation problem, in which the sensors need to estimate the initial external time. This work is further relevant to the problem of collective approximation of environmental values by biological groups. 

Unlike most works on clock synchronization that assume static networks, this paper focuses on an extreme case of highly dynamic networks.  
Specifically, we assume  a non-adaptive scheduler adversary that dictates in advance an arbitrary, yet {\em independent}, meeting pattern. Such meeting patterns fit, for example, with short-time scenarios in highly dynamic settings, 
where each sensor interacts with only few other arbitrary sensors.

We  propose an extremely simple clock synchronization algorithm that is based on weighted averages, and prove that its performance on any given independent meeting pattern is highly competitive with that of the best possible algorithm, which operates without any resource or computational  restrictions, and knows the meeting pattern in advance. In particular, when all distributions involved are Gaussian, the performances of our scheme coincide with the optimal performances. Our proofs rely on an extensive use of the concept of Fisher information. We use the Cram\'er-Rao bound and  our definition of a \emph{Fisher Channel Capacity} to quantify information flows and to obtain lower bounds on collective performance. This opens the door for further  rigorous quantifications of information flows within collaborative sensors.\\
\end{abstract}




\end{titlepage}



\section{Introduction}

\subsection{\large Background and Motivation}\parskip-0cm

 Representing and  communicating information is a main interest of theoretical distributed computing \cite{Peleg:book}. However, such studies often seem disjoint from what may be the largest body of work regarding coding and communication: Information theory \cite{Cover,Shannon}. Perhaps the main reason  for this stems from the fact that theoretical distributed computing studies are traditionally concerned with  noiseless models of communication, in which the content of a message that passes from one node to another is not distorted. 
%
%
This reliability in transmission relies on an implicit assumption that error-corrections is guaranteed by a lower level protocol that is responsible for implementing  communication. Indeed, when bandwidth is sufficiently large, one can encode a message with a large number of error-correcting bits in a way that makes communication noise practically a non-issue. \parskip0.1cm

In some distributed scenarios, however,   distortion in communication~is unavoidable. 
One  example concerns the classical problem of {\em clock synchronization}, which has attracted a lot of attention from both theoreticians in distributed computing~\cite{Attiya,Lynch,Lenzen2,Boaz2,Srikanth}, as well as  practitioner engineers \cite{Chaudhari2,Elson,Elson2,Fetzer,Jeske,Sichitiu,Solis}, see  \cite{survey, survey2,Lenzen3,Qasim} for comprehensive surveys. In this problem,   processors need to synchronize their internal clocks  (either among themselves only or with respect to a global time reference) relying on relative time  measurements between clocks. Due to unavoidable  unknown delays in communication, such measurements are inherently noisy. Furthermore, since the source of the noise is  the delays, error-correction does not seem to be of any use for reducing the noise.  
 The situation becomes even more complex when processors are mobile, preventing them from  reducing errors by averaging repeated measurements to the same processors, and from contacting  reliable processors. Indeed, the clock synchronization problem is particularly challenging in the  context of wireless sensor networks and ad hoc networks which are  typically formed by autonomous, and often mobile, sensors without central control. 


Distributed computing models which include noisy communication call for a rigorous comprehensive study that employs information theoretical tools.
Indeed, a recent trend in the engineering community is to 
view the clock synchronization problem from a signal processing point of view, and adopt tools from information theory (e.g., the Cram\'er-Rao bound) to bound the affect/impact of inherent noise  \cite{Chaudhari,Chaudhari2,Jeske,Leng}, see \cite{Qasim} for a survey. However, this perspective has hardly received any attention by theoreticians in distributed computing that mostly focused on worst case message delays \cite{Attiya,Lynch,Lenzen2,Welch}, which do not seem to be  suitable for information theoretic considerations. In fact, very few works on clock synchronisation consider a system with random  delays and analyse it following a rigorous theoretical distributed algorithmic type of analysis. An exception to that is the work of Lenzen el al.  \cite{Lenzen},  but also that work does not involve information theory. In this current paper, 
we study  the clock synchronization problem through the purely theoretical distributed algorithmic perspective while adopting the signal processing and information theoretic point of view. In particular, we adopt tools from Fisher Information theory \cite{Stam1959,Zamir1998}.

We consider the  {\em external} version of the  problem \cite{Cristian,Elson2,Fetzer,Mills,Boaz2,survey2} in which processors (referred to as sensors hereafter) collaborate in order to synchronize their  clocks with an external {\em global clock}. Informally, sensors initially obtain inaccurate estimates\footnote{Traditional protocols like NTP \cite{NPT} and TEMPO \cite{Tempo}
use an external standard like GPS (Global Positioning System) or UTC (Universal~Time) to synchronize 
networks. However, the use of of such systems poses a high demand for energy which is usually
undesired in  sensor networks. Hence, works in sensor networks typically assume that one {\em source} processor obtains  an accurate~estimate of the global time reference and essentially governs the synchronization of the rest of the  sensors \cite{Boaz2}. Here, we generalize this framework by assuming that each processor may  initially have a different estimate quality of the global time reference, and our goal is to investigate what can be achieved given the qualities of initial estimations.}
of a global (external) time $\tau^*\in \mathbb{R}$ reference,
and subsequently collaborate to align their internal clocks to be as close as possible to the external clock. To this end, sensors  communicate through uni-directional pairwise interactions that include inherently {\em noisy measurements} of the relative deviation between their internal clocks and, possibly, some complementary information. 
To~focus~on the problems occurred by the initial inaccurate estimations of $\tau^*$ and the noise in the communication~we~restrict~our attention to {\em drift-free} settings \cite{Attiya,Lynch}, in which all clocks tick at the same rate. This setting essentially reduces~the~problem to the problem of estimating $\tau^*$. See, e.g.,~\cite{Gubner,Xiao,Detection} for works on estimation in the engineering community. 

With very few exceptions that effectively deal with dynamic settings  \cite{Dolev,kuhn}, almost all works on clock synchronization (and estimation) considered static networks. 
Indeed, the construction of efficient clock synchronization algorithms for dynamic networks is considered as a very important and challenging task\footnote{For example, dynamic meeting patterns prevent the use of classical external clock synchronization algorithms (e.g., \cite{NPT,Boaz2})  that are based on one or few  {\em source} sensors that obtain accurate estimation of the global time and govern the synchronization of other sensors.} \cite{survey, survey2}.
This paper addresses this challenge by considering highly dynamic networks in which sensors have little or no control on who they interact with. Specifically, we assume a non-adaptive scheduler adversary that dictates in advance a meeting-pattern for the sensors. However, the adversary we assume is not unlimited. Specifically, for simplicity, in this initial work we restrict the adversary to provide {\em independent-meeting patterns} only, in which it is guaranteed that whenever a sensor views another sensor, their transitive histories are disjoint\footnote{Another informal way to view such patterns is that they guarantee that, given the global time, whenever a sensor views another sensor, their local clocks are independent; see Section \ref{sec:model} for a formal definition.}. 
Although they are not very good representatives of communication in static networks, 
independent meeting patterns 
fit  well with highly stochastic communication patterns during short-time scales, in which each sensor observes only few other arbitrary sensors (see more discussion in Section~\ref{sec:model}). 
Given such an adversarial meeting-pattern, we are concerned with minimizing the deviation of each internal clock from the global time.

 

As our objective is to model small and simple sensors, we are interested in algorithms that employ elementary internal computations and
economic use of communication.
We use competitive analysis to evaluate the performances of algorithms, comparing them to the best possible  algorithm that knows the whole meeting pattern in advance and operates under the most liberal version of the model that allows for unrestricted resources in terms of memory and communication capacities, and individual computational ability.


\bigskip

\subsection{\large Our contribution}\label{results}
\paragraph{Lower bounds on optimal performance.} 
We first consider algorithm $\OPT$, the best possible algorithm operating on the given independent meeting pattern. We note that specifying $\OPT$ seems challenging, especially since we do not assume a prior distribution on the starting global time, and hence the use of Bayesian statistics seems difficult. Fortunately, for our purposes, we are merely interested in lower bounding the performances on that algorithm. We achieved  that by 
 relating the smallest possible variance of a sensor at a given time  to the largest possible  \emph{Fisher Information (FI)} of the sensor at that time. \parskip0.1cm
%
%
%
 This measure 
  quantifies the sensor's current  knowledge regarding the relative deviation between its local time and the global time.
  We provide a recursive formula to calculate~$J_a$, the FI at  sensor $a$, for any sensor $a$. Specifically, initially, the FI at a sensor is the Fisher information in the distribution family governing its initial deviation from the global time (see Section~\ref{pre:fisher} for the formal definitions). When sensor~$a$ observes sensor $b$, the FI at $a$ after this observation (denoted by $J'_a$) satisfies: 
  \begin{equation}\label{eq:fish}
  J'_a\leq J_a + \frac{1}{\frac{1}{J_b} + \frac{1}{J_{N}}},
   \end{equation}
  where $J_{N}$ is the Fisher Information in the noise distribution related to the observation.  To obtain this formula we  prove   
a  generalized version of the {\em Fisher information inequality} \cite{Stam1959,Zamir1998}. 
Relying on the {\em Cram\'er-Rao bound} \cite{Cover}, this formula is then used  to bound the corresponding variance under algorithm~$\OPT$.
Specifically, the variance of the internal clock of sensor $a$ is at least ${1}/{J_a}$.

Equation \ref{eq:fish} provides immediate bounds on the convergence time. Specifically,  
the inequality  sets a bound~of~$J_{N}$ for the increase in the  {\em FI} per interaction.
In analogy to Channel Capacity as defined by Shannon \cite{Cover} we term this upper bound  as the {\em Fisher Channel Capacity}.
Given small $\epsilon>0$, we define the convergence time $T(\epsilon)$  as the minimal number of observations required by the typical sensor   until  its variance drops below  $\epsilon^2$  (see Section~\ref{pre:fisher} for the formal definition).  Let $J_0$ denote the median initial Fisher Information of sensors.
Based on the Fisher Channel Capacity we  prove the following.
\begin{theorem}\label{eq:time} Assume that $J_0\ll 1/\epsilon^2$ for some small $\epsilon>0$, then
$T(\epsilon)~\geq ({\frac{1}{\epsilon^2}-J_0})/{J_{N}}\approx {1}/{\epsilon^2 J_{N}}.$
\end{theorem}

\paragraph{A highly competitive elementary algorithm.}
We  propose a simple clock synchronization algorithm
and prove~that its performance on any given independent meeting pattern  is highly competitive with that of the optimal one. That is, estimations of global time at each sensor remain unbiased throughout the execution and the variance at any given time is $\Delta_0$-competitive with the best possible variance, where $\Delta_0$ is initial Fisher-tightness (see definition in Section \ref{pre:fisher}. In contrast to the optimal algorithm that may be based on transmitting complex functions in each interaction, and on performing complex internal computations, our simple algorithm is based on far more basic rules. First,  transmission is  restricted to a single {\em accuracy} parameter. Second,  using the noisy measurement of deviation from the observed sensor, and the accuracy  of that sensor, the observing sensor updates its internal clock and accuracy parameter by  careful, yet elementary, weighted-averaging~procedures.


Our weighted-average algorithm  is designed to maximize the flow of Fisher information in interactions. This is proved by showing that the accuracy  parameter is, at all times, both representative of the reciprocal of the sensor's variance  and close to the Fisher Information upper bound. 
In short, we proved the following.
\begin{theorem}\label{thm:main}
There exists a simple weighted-average based clock synchronization algorithm which  is $\Delta_0$-competitive (at any sensor and at any time).
\end{theorem}
Two important corollaries of Theorem \ref{thm:main} follow directly from the definition of the  initial Fisher-tightness $\Delta_0$.
\begin{corollary}
If the  number of distributions involved is a constant (independent of the number of sensors), then  our algorithm is $O(1)$-competitive (at any sensor and at any time).
\end{corollary}

\begin{corollary}
If all  distributions involved are  Gaussians, then the performances of our algorithm (in terms of the variances) coincide with the optimal one, for each sensor and at any time.
\end{corollary}
 

We note that our algorithm does not require the use of sensor identities and can thus be also employed in {\em anonymous} networks \cite{Ang,Breathe}, yielding the same performances.

\bigskip

\section{\large Preliminaries}\label{preliminaries}\label{subset:fishfish}\parskip0.1cm

\subsection{The Model}\label{sec:model} \parskip0.1cm
We consider a collection of $n$ 
sensors that collaborate in order to synchronize their internal clocks with an external global clock reference. 
We consider a set $\cF$ of  sufficiently smooth (see definition in Section~\ref{pre:fisher}),  probability density distributions ($pdf$) centered at zero. One specific distribution among the $pdf$s in~$\cF$ is the {\em noise} distribution, referred to as $N(\eta)$. Each sensor $a$ is associated with a distribution $\Phi_a(x)\in \cF$ which governs the initialization deviation of its internal clock from the global time as described in the next paragraph. Depending on the specific model, we assume that sensor $a$ knows various properties of~$\Phi_a$. In the most restricted version we consider, sensor~$a$ knows only the variance of $\Phi_a$ and in the most liberal version, $a$  knows the full description of $\Phi_a$. Execution is initiated when the global time is some $\tau^*\in \mathbb{R}$, chosen by an adversary.
  
Two important cases are (1) when $\cF$ contains a constant number of distributions (independent of the number of sensors) and (2) when all distributions in $F$ are Gaussian. Both cases serve as reasonable assumptions for realistic scenarios. For the former case we shall show asymptotically optimal performances and for the latter case we will show strict optimal (non-asymptotical) performance.
 
\paragraph{Local clocks.}
Each sensor~$a$ is initialized with a local clock  $\ell_a(0)\in \mathbb{R}$,  randomly chosen according to  $\Phi_a(x-\tau^*)$, independently of all other sensors. That is, as $\Phi_a(x)$ is centred around zero, the initial local time $\ell_a(0)$ is distributed around $\tau^*$, and this distribution is governed by $\Phi_a$. We stress that sensor $a$ does not know the value $\tau^*$ and from its own local perspective the execution started at time $\ell_a(0)$.  Sensors rely on both social interactions and further environmental cues\footnote{In order for the model to include environmental cues, one or more of the sensors can be taken to represent the global clock. The initial times of these sensors are chosen according to highly concentrated distributions, $\Phi_a$,  around $\tau^*$ and  remain fixed thereafter.}
to improve their estimates of the global time. In between such events  sensors are free to  perform ``shift'' operations to adjust their local~clocks. 
%
To focus on the problems occurred by the initial inaccurate estimations of $\tau^*$ and the noise in the communication we restrict our attention to {\em drift-free} settings \cite{Attiya,Lynch}, in which all clocks tick at the same rate, consistent with the global time.

\paragraph{Opinions.}
The drift-free assumption reduces the external clock-synchronization problem to the problem of estimating $\tau^*$. Indeed,  recall that local clocks are initialized to different values but progress at the same rate.  Because sensor $a$ can keep the precise time since the beginning of the execution, its deviation from the global time can be corrected had it known the difference between, $\ell_a(0)$, the initial local clock of $a$, and $\tau^*$, the global time when the execution started. Hence, one can view the goal of sensor $a$ as estimating $\tau^*$. That is, without loss of generality, we may assume that all shifts performed by sensor $a$ throughout the execution are shifts of its initial position $\ell_a(0)$ aiming to align it to be as close as possible to $\tau^*$. Taking this perspective, we associate with each sensor  an {\em opinion} variable $x_a$, initialized to $x_a(0):=\ell_a(0)$,  and the goal of $a$ is to have its opinion be as close as possible to
$\tau^*$. 
We view the opinion $x_a$ as an {\em estimator} of $\tau^*$, and note that initially, due to the properties of $\Phi_a$, this estimator is unbiased, i.e., $\mean(x_a(0)-\tau^*)=0$. It is required that at any point in the execution, the opinion $x_a$  remains an unbiased estimator of $\tau^*$, and the goal of $a$ is to  minimize its Mean Square Error (MSE).

Due to this simple relation between internal clocks and opinions, in the remaining of this paper, we shall adopt the latter perspective and 
concern ourselves only with optimizing the opinions of  sensors as estimators for $\tau^*$, without discussing further the internal clocks.

\paragraph{Rounds.}
For simplicity of presentation, we assume that the execution proceeds in discrete steps, or rounds. We stress however that the rounds represent the order in which communication events occur (as determined by the meeting-pattern, see below), and do not necessarily correspond to the actual time. Given an algorithm $A$, the opinion maintained by the algorithm at round $t$ (where $t$ is a non-negative integer) at sensor $a$ is denoted by $x_a(t,A)$.  As mentioned, the algorithm  aims to keep this value as close as possible to~$\tau^*$. When $A$ is clear form the context, we may omit writing it and use the term $x_a(t)$ instead. 

In each round $t\geq 1$, each sensor may first choose to shift (or not) its opinion, and then, if specified in the meeting pattern, it  observes another specified sensor, thus obtaining some information. To summarize, in each round, a sensor executes the following consecutive actions:
(1)~Perform internal computation;
(2) Perform an opinion-shift: $x_a(t)=x_a(t-1)+\Delta(x)$;
and (3)
Observe (or not) another sensor.
For simplicity, all these three operations
are assumed to occur instantaneously, that is, in zero time.

\paragraph{Mobility and adversarial independent meeting patterns.} \parskip0.1cm 
In cases where sensors are embedded in a Euclidian space, distances between positioning of sensors may impact the possible interactions. To account for physical mobility, and be as general as possible, we assume that  an oblivious adversary controls the meeting pattern. That is, the adversary decides (before the execution starts), for each round, which sensor observes which other sensor. 

A model that includes an unlimited adversary that controls the meeting pattern is in some sense too general\footnote{For example, 
in our model we assume that sensors are anonymous but we compare such algorithm to the best possible algorithm that knows the identities of sensors and the whole meeting pattern in advance. In such case,  an arbitrary interaction pattern can match all sensors such that interactions occur only within  pairs. As the sensors are anonymous, they cannot distinguish this case from  other, more uniform, meeting patterns and hence, cannot be expected to act as efficiently as algorithms with identified sensors.
Some limitations on the adversarial interaction network are therefore required.}.
In this preliminary work on the subject, we restrict the adversary to provide only  {\em independent} meeting patterns, in which the set of sensors in the transitive history of each observing sensor is disjoint from the one of the observed sensor. As indicated by this work, the case of independent meeting patterns is already complex. We leave it to future work to handle dependent meeting patterns. 

Formally,
given a pattern of meetings~$\cP$, sensor $a$ and round $t$, we first  define the~set~of~{\em relevant} sensors of  $a$ at time~$t$, denoted by  $\cR_a(t,\cP)$.
 At time zero,  we define $\cR_a(0,\cP):=\{a\}$, and at round $t$, $\cR_a(t,\cP):=\cR_a(t-1,\cP)\cup \cR(b,t-1,\cP)$ if $a$ observes $b$ at time~$t-1$ (otherwise $\cR_a(t,\cP):=\cR_a(t-1,\cP)$). A meeting pattern $\cP$ is called {\em independent}  if whenever some sensor~$a$ observes a sensor $b$ at some time~$t$,   then $\cR_a(t-1,\cP)\cap \cR(b,t-1,\cP)=\emptyset~.$
 Note that an 
 independent meeting pattern guarantees that given $\tau^*$, the internal clocks of two interacting sensors  are independent. However,  given $\tau^*$ and the internal clock  of $a$, the internal clock of $b$  and the relative time measurement between them are dependent (this point is explained in further details in Section \ref{difficulties}).

Note that independent-meeting patterns are not very good representatives of communication in static networks\footnote{Indeed, in such patterns a sensor will not contact the same sensor twice, which contradicts many natural communication schemes in static networks. We note, however, that in some cases, a sequence of multiple consecutive observations between sensors can be compressed into a single  observation of higher accuracy thus reducing the dependencies between observations, and possibly converting a dependent meeting pattern into an independent one. For example, if sensors have unique identities and sensor $a$ observes sensor $b$ several times is a row, and it is guaranteed that sensor $b$ did not change its state during these observations, then these observations can be treated by  $a$ as a single, more accurate, observation of  $b$.}. On the other hand, independent meeting patterns fit  well with highly stochastic short-time scales communication patterns, in which each sensor observes only few other arbitrary sensors.  In this sense, such patterns can be considered as representing an extreme case of dynamic systems.
 
\parskip0.1cm

Because sensors have no control of when their next interaction will occur, or if it will occur at all, we require that estimates at each sensor be as accurate as possible at {\em any} point in time. This requirement is stronger than the liveness property that is typically required from distributed algorithms \cite{Lamport}.

\paragraph{Convergence time.} 
Consider a meeting pattern $\cP$. Given small $\epsilon>0$, the {\em convergence time} $T(\epsilon)$ of an algorithm $A$ is defined as the minimal number of observations made by the typical sensor until its variance is less than $\epsilon^2$. More formally, let $\rho$ denote  the first round when we have more than half of the population satisfying
 $\var(X_a(t,A))<\epsilon^2$. For each sensor $a$, let $R(a)$ denote the number of observations made by $a$ until time $\rho$. The convergence time $T(\epsilon)$ is defined as the median of $R(a)$ over all sensors $a$. Note that $T(\epsilon)$ is a lower bound on $\rho$, since $\rho \geq R(a)$ for every sensor $a$.


 %


\paragraph{Communication.}\parskip0.1cm
We assume that sensors are anonymous and hence, in particular, they do not know who they observe. Conversely, for the sake of lower bounds, we allow a much more liberal setting, in which sensors have unique identifiers and know who they interact with.

When a sensor $a$ observes another sensor $b$ at some round $t$,
the information transferred in this interaction contains a {\em passive} component and, possibly, a complementary  {\em active} one. The passive component is a noisy relative deviation measurement between their opinions:
$$\tilde{d}_{ab}(t)=x_b(t)-x_a(t)+\eta,$$ where the additive noise term, $\eta$, is  chosen from the noise probability distribution $N(\eta)\in \cF$ whose variance is known to the sensors. (Note that this measurement is equivalent to the relative deviation measurement between the sensors' current local times because all clocks tick at the same pace.)


\subsection{Elementary algorithms}
Our reference for evaluating performances is  algorithm $\OPT$ which operates under the most liberal version of our model, which carries no restrictions on memory, communication capacities or internal computational power, and  provides the best possible estimators at any sensor and at any time (we further assume that sensors acting under  $\OPT$ know the meeting pattern in advance).
In general,  algorithm $\OPT$ may use complex calculations over very wasteful memories that include detailed distribution density functions, and possibly, accumulated measurements.
Our  main goal is to identify an algorithm whose performance  is highly competitive with that of $\OPT$ but wherein communication and memory are economically used, and the local computations simple.
Indeed, when it comes to applications to tiny and limited processors, simplicity and  economic use of communication~are~crucial~restrictions.

 An algorithm is called {\em elementary} if the internal state  of
each sensor $a$ is some real\footnote{We assume  real numbers for simplicity. It seems reasonable to assume that when sufficiently accurate approximation is stored instead of the real numbers similar results could be obtained.} number $y_a\in \mathbb{R}$, and, more importantly, the  internal computations that a sensor  can perform  consist of a constant number 
of  basic arithmetic operations, namely: addition, subtraction, multiplication, and~division.

 \parskip0.1cm

\subsection{Competitive analysis}

Fix a finite family~$\cF$ of  smooth $pdf$'s centered at zero (see the definition for smoothness in the next paragraph), and fix an assignment of a distribution $\Phi_a\in \cF$ to each sensor $a$. For an algorithm $A$ and an independent meeting pattern  $\cP$, let $X_a(t,A,\cP)$  denote the random variable indicating the opinion of sensor $a$ at round $t$. Let $\mean(X_a(t,A,\cP))$ and $\var(X_a(t,A,\cP))$ denote, respectively, the mean and variance of $X_a(t,A,\cP)$,
where these are taken over all possible random initial opinions, communication errors, and possibly, coins flipped by the algorithm. Note that the unbiased assumption requires that $\mean(X_a(t,A,\cP))=\tau^*$.
An algorithm~$A$ is called $\lambda$-competitive, if for {\em any} independent pattern of meetings $\cP$, {\em any} sensor $a$, and at {\em any} time~$t$, we have:
$
\var(X_a(t,A,\cP))\leq \lambda\cdot \var(X_a(t,\OPT,\cP)).
$

\subsection{Fisher information and the Cram\'er-Rao bound}\label{pre:fisher}

The Fisher information is a standard way of evaluating the amount of information that a set of random measurements holds about an unknown parameter $\tau$ of the distribution from which these measurements were taken.
We provide here some definitions for this notion; for more information   the reader may refer to
\cite{Cover,Zamir1998}.

A single variable probability distribution function ($pdf$) $\Phi$ is called  {\em smooth} if it satisfies the following conditions, as stated by Stam \cite{Stam1959}:
(1) $\Phi(x)>0$ for any $x\in \mathbb{R}$,
(2) the derivative $\Phi'$ exists, and
(3) the integral $\int \frac{1}{\Phi(y)}(\Phi'(y))^2dy$ exists, i.e., $\Phi'(y)\rightarrow 0$ rapidly enough for $|y| \rightarrow \infty$. Note that, in particular, these conditions  hold for natural distributions such as the Gaussian distribution. Recall that we consider a finite set $\cal{F}$ of smooth one variable $pdf$s, one of them being the noise distribution $N(\eta)$, and all of which~are~centered~at~zero.

For a smooth $pdf$ $\Phi$, let $J^\tau_{\Phi}:=\int \frac{1}{\Phi(y)}(\Phi'(y))^2dy$ denote the Fisher information in the parameterized family $\{(\Phi(x,\tau)\}_{\tau\in\mathbb{R}}=\{(\Phi(x-\tau)\}_{\tau\in\mathbb{R}}$  with respect to $\tau$.
In particular, let  $J_{N}=J^\tau_N$ denote the Fisher information in the parameterized family $\{N(\eta-\tau)\}_{\tau\in\mathbb{R}}$.
More generally, consider a multi-variable $pdf$ family $\{(\Phi(z_1-\tau, z_2 \ldots z_k))\}_{\tau\in\mathbb{R}}$ where~$\tau$ is a translation parameter.
The Fisher information in this family
with respect to~$\tau$ is defined as:
$$
J_{\Phi}^{\tau}=\int\frac{1}{\Phi(z_1-\tau, z_2 \ldots z_k)}~\left[~\frac{d\Phi(z_1-\tau, z_2 \ldots z_k)}{d\tau}~\right]^2~ dz_1, dz_2 \ldots dz_k~
~~\mbox{(if the integral exists)}$$ 
As previously noted \cite{Zamir1998},   since $\tau$ is a translation parameter, Fisher information is both unique (there is no freedom in choosing the parametrization) and independent of  $\tau$.

The Fisher information derives its importance by association with the Cram\'er-Rao inequality \cite{Cover}. This inequality lower bounds  the variance of the best possible estimator of~$\tau^*$  by the reciprocal of the Fisher information that corresponds to the random variables on which this estimator is~based.
\begin{theorem}\label{thm:Cramer} {\bf [The Cram\'er-Rao inequality]}
Let $\hat{X}$ be any unbiased estimator of $\tau^*\in\mathbb{R}$ which is based on a multi-variable sample $\bar{z}=(z_1, z_2 \ldots z_k)$ taken from~$\Phi(z_1-\tau^*, z_2 \ldots z_k)$. Then
$\var (\hat{X})\geq {1}/{J_{\Phi}^\tau} .$
\end{theorem}

\paragraph{Initial Fisher-tightness :}
To define the initial Fisher-tightness  parameter $\Delta_0$, we first define the {\em Fisher-tightness } of a single variable smooth distribution $\Phi$ centered at zero, as $\Delta(\Phi)= \var(\Phi)\cdot J^\tau_\Phi~.$ Note that, by the Cram\'er-Rao bound, $\Delta(\Phi)\geq 1$ for any such distribution $\Phi$. Moreover, equality holds if  $\Phi$ is Gaussian \cite{Cover}.
Recall that  $\cF$ is the finite collection of  the smooth distributions  containing the distributions $\Phi_a$ governing the initial opinions of sensors.
The {\em initial Fisher-tightness} $\Delta_0$ is the maximum of the Fisher-tightness  over all distributions in  $\cF$ and the noise distribution. Specifically, let  $\Delta_0=\max\{ \Delta(\Phi)\mid \Phi\in\cF\}.$ 
Two important observations are: 
\begin{itemize}
\item If $\cF$ contains a constant number of distributions then $\Delta_0$ is a constant.
\item
If the distributions in~$\cF$  are all  Gaussians then $\Delta_0=1$.
\end{itemize}

%
%

\section{\large Technical difficulties}\label{difficulties}
It is known for a single sensor, one can associate weights to samples so that a weighted-average procedure can fuse them optimally  \cite{McNamara1987a}.  The proof therein relies on the assumption  that all probability distributions are  Gaussians whose functional forms, indeed, depend on their second moments only.
Our setting is more complex, since it includes arbitrary differentiable $pdf$'s and multiple distributed sensors whose relative opinions constantly change.\parskip0.1cm

The extension to multiple mobile sensors adds another dimension to the problem. One difficulty lies in the fact that the partial knowledge held by each sensor is relative ({\em e.g.}, an estimation of the deviation  between the sensor's opinion and $\tau^*$) and hence may require the sensors to carefully fuse  perspectives other than their own. This difficulty is enhanced, as the sensors  constantly shift their opinions.
Indeed, for elementary algorithms, where memory is restricted to a single parameter, storing the sum of previous shifts in the memory of a sensor is possible, but could drastically limit the degrees of freedom for encoding other information. On the other hand, without encoding previous shifts, it is not clear how sensor $a$ should treat information it had received from $b$.

In addition, compression of memory and communication appears to be detrimental. Indeed,  maintaining and communicating highly detailed memories can, in some cases, significantly improve a sensor's assessment of the target value. However, maintaining a high degree of detail requires storing an arbitrary number of $pdf$ moments which may grow with every interaction. Hence,  it is not clear how to compress the information into  few meaningful parameters while avoiding the accumulation of errors and runaway behavior.

Several technical difficulties arise when attempting to bound the performances of different algorithms. In natural type of algorithms, sensors' memories can be regarded as maintaining $pdf$s that summarize their knowledge regarding their deviation from the  target value $\tau^*$.
One of the analysis difficulties corresponds to the fact that the $pdf$ held by a sensor at round $t$ depends on many previous deviation measurements in a non-trivial way, and hence the variance of a realization of the $pdf$ does not necessarily correspond to the variance of the sensors' opinion, when taking into account all possible realizations of all measurements. Hence, one must regard each $pdf$ as a multi-variable distribution. A second problem has to do with dependencies. The independent meeting pattern guarantees that the memory $pdf$'s of two interacting sensors are independent, yet, given the $pdf$ of the observing sensor, the $pdf$ of the observed sensor and the deviation measurement become dependent. Such dependencies make it difficult to track the evolution of a sensor's accuracy of estimation over time. Indeed, to tackle this issue, we had to extend the Fisher information inequality \cite{Stam1959,Zamir1998,Rioul}  to a multi-variable dependent convolution case. \parskip0.2cm

\bigskip

\section{Lower bounds on the variance of $\OPT$}
\parskip0.1cm
In this section we provide lower bounds on the performances of algorithm $\OPT$ over a fixed independent pattern of meetings $\cP$. Note that we are interested in bounding the performances of $\OPT$  and not in specifying its instructions. Identifying the details of $\OPT$ may still be of interest, but it is beyond the scope of this paper. 

For simplicity of presentation, we assume that the rules of $\OPT$  are deterministic. We note, however, that our results can easily be extended to the case that $\OPT$ is probabilistic. For simplicity of notations, since this section deals only with algorithm $\OPT$ acting over $\cP$, we use variables, such as the opinion $X_a(t)$ and the memory~$Y_a(t)$ of sensor $a$, without parametrizing them by neither  $\OPT$ nor by~$\cP$.

Under algorithm $\OPT$, we assume that each sensor holds initially not only the variance of $\Phi_a$, but the precise functional form of the distribution $\Phi_a$ (recall, $\Phi_a$ is centered at zero). In addition, we assume that sensors have unique identifiers and that each sensor knows the whole pattern $\cP$ in advance. Moreover, we assume that each sensor $a$ knows for each other sensor $b$, the $pdf$ $\Phi_b$ governing $b$'s initial opinion. All this information is stored in one designated part of the memory of $a$.

Since $\OPT$ does not have any bandwidth constrains,  we may assume, without loss of generality,  that whenever some  sensor $a$ observes another sensor $b$, it obtains the whole memory content of $b$. Since  $\OPT$ is deterministic, its previous opinion-shifts can be extracted from its interaction history, which is, without loss of generality, encoded in its memory\footnote{In case $\OPT$ is probabilistic, previous shifts can be extracted from the memory plus the results of coin flips which may be encoded in the memory of the sensor as well.}. Hence, when sensor~$a$ observes sensor $b$ at some round $t$, and receives $b$'s memory together with the noisy measurement $\tilde{d}_{ab}(t)=x_b(t)-x_a(t)+\eta$, sensor $a$ may extract all previous opinion-shifts
of both itself and~$b$, treating the measurement $\tilde{d}_{ab}(t)$ as a  noisy measurement of the deviation between the initial opinions, i.e., $\tilde{d}_{ab}(0)=x_b(0)-x_a(0)+\eta$. In other words,  to understand the behavior of $\OPT$ at round $t$, one may assume
 that sensors never shift their opinions until round $t$, when they use all memory they gathered to shift their opinion in the best possible manner\footnote{
This observation implies, in particular, that previous opinion-shifts of sensors do not affect subsequent estimators in a way that may cause a conflict (a conflict may arise, e.g., when optimizing one sensor at one time necessarily makes estimators at another sensor, at a later time, sub-optimal), hence algorithm $\OPT$ is well-defined.}.
It follows that apart from the designated memory part that all sensors share, the memory~$M_a(t)$ of sensor $a$ at round $t$ contains the initial opinion $X_a(0)$ and   a collection $Y_a(t-1):=\{\tilde{d}_{bc}(0)\}_{bc}$ of relative deviation measurements between initial opinions.
That is, $M_a(t)=(X_0(t),Y_a(t-1)).$ This multi-valued memory  variable $M_a(t)$ contains all the information available to $a$ at round $t$. In turn, this information is used by the sensor to obtain its opinion $X_a(t)$ which
is required  to serve as an unbiased estimator of $\tau^*$.

\subsection{\large The Fisher Information of sensors}\label{sub:relative}
We now define the notion of the Fisher Information associated with a sensor $a$ at  round $t$. This definition will be used to bound from below the variance of  $X_a(t)$ under algorithm $\OPT$.

Consider  the  multi-valued memory  variable $M_a(t)=(X_0(t),Y_a(t-1))$ of sensor $a$ that at  round $t$.
Note that $Y_a(t-1)$ is independent of $\tau^*$. Indeed, once the adversary decides on the value  $\tau^*$, all sensors' initial opinions are chosen with respect to $\tau^*$. Hence, since sensors' memories contains only relative deviations between opinions, the memories by themselves do not contain any information regarding~$\tau^*$. 
In contrast, given~$\tau^*$, the random variables $Y_a(t-1)$ and $X_a(0)$ are, in general, dependent. 
Furthermore, in contrast to $Y_a(t-1)$, the value of  $X_a(0)$ depends on $\tau^*$, as it is chosen according to $\Phi_a(x-\tau^*)$. Hence,
${M}_{a}(t)$ is distributed according to a $pdf$ family $\{(m_{a}(t),\tau)\}$ parameterized by a translation parameter $\tau$.
Based on $M_{a}(t)$, the sensor produces  an unbiased estimation ${X}_{a}(t)$ of~$\tau^*$, that is,
it should hold that:
$
\mean({X}_{a}(t)-\tau^*)=0,
$
where the mean is taken with respect to the distribution of the random multi-variable  ${M}_{a}(t)$.\parskip0.2cm

\noindent{\bf Definition:} The {\em Fisher Information (FI)} of sensor $a$ at  round  $t$, termed   $J_a(t)$,  
is the the Fisher information in the parameterized family $\{(m_{a}(t),\tau)\}_{\tau\in\mathbb{R}}$ with respect to $\tau$.


By the Cram\'er-Rao bound, the variance  of any unbiased estimator used by the sensor~$a$ at round $t$ is bounded from below by
the reciprocal of the {\em FI} of sensor~$a$ at that time. That is, we have:
\begin{lemma}\label{for:fish}
$\var(X_a(t))\geq {1}/{J_a(t)}$.
\end{lemma}


\subsection{\large An upper bound on the Fisher Information $J_a(t)$}\parskip0.2cm
Lemma \ref{for:fish} implies that
lower bounds on the variance of the opinion of a sensor  can  be obtained by bounding from above the corresponding  {\em FI}.
To this end,
 our next goal  is to prove the following recursive inequality.
\begin{theorem}\label{lem:fishernoise}
The {FI} of sensor $a$ under algorithm $\OPT$ satisfies:   $J_a(t+1)\leq J_a(t) + {1}/({\frac{1}{J_b(t)} + \frac{1}{J_{N}}}).$
\end{theorem}
\begin{proof}
Consider the case that at  round $t$, sensor $a$ observes sensor $b$. After the interaction, the random multi-variable   $Y_a(t)$ is composed of:
(1)
the random variable $\tilde{D}_{ab}(0):=X_b(0)-X_a(0)+N$,  corresponding to  the noisy deviation measurement between the initial opinions of $a$ and $b$, and
(2)
the relative deviation measurements in both $Y_a(t-1)$ and $Y_b(t-1)$.
We now aim at calculating the FI $J_{a}(t+1)$ available to sensor $a$ at time $t+1$, with respect to the parameter~$\tau$. This is
the FI with respect to~$\tau$,
 in the multi-variable ${M}_{a}(t+1)=(X_a(0),Y_a(t))=(X_a(0),\tilde{D}_{ab}(0),Y_a(t-1), Y_b(t-1))$, where $X_a(0)$ is distributed according to $\Phi_a(x-\tau^*)$. Taking $\tilde{X}_b(0)=X_a(0)+ \tilde{D}_{ab}(0)=X_b(0)+N$, this latter FI becomes the same as the FI in the random variables: $(X_a(0),\tilde{X}_b(0), Y_a(t-1), Y_b(t-1))$.
Since the meeting pattern is independent, then given the environment value  $\tau^*$,  the random multi-variable  $(X_a(0), Y_a(t-1))$ is  independent  of the random multi-variable   $(\tilde{X}_b(0), Y_b(t-1))$. By the additivity property of the Fisher information with respect to independent random multi-variables (see \cite{Stam1959}),
the {\em FI} $J_{a}( t+1)$   therefore equals the {\em FI} $J_a(t)$ (which is the FI in the random multi-variable $(X_a(0), Y_a(t-1))$)  plus the  FI $\tilde{J}_{b}(t)$ in the random multi-variable $(\tilde{X}_b(0), Y_b(t-1))$, both with respect to $\tau$. That is,
we have:
\begin{equation}\label{Eq:1}
J_a( t+1,A)= J_a(t)+\tilde{J}_{b}(t).
\end{equation}
Let us now focus on the rightmost term in  Equation \ref{Eq:1} and calculate $\tilde{J}_{b}(t)$.
Given that the target value is some~$\tau$, the distribution  of $(\tilde{X}_b(0),Y_b(t-1))$  can  be described by the following convolution:
\begin{equation}\label{Eq:3}
f_{\tilde{X}_b(0),Y_b(t-1)}[(\tilde{x}_b(0),y_b(t-1))\mid \tau]=\int f_{X_b(0),Y_b(t-1)}[\tilde{x}_b(0)-\eta, y_b(t-1)\mid \tau] ~N(\eta)~d\eta.
\end{equation}
Observe that the right hand side of Equation \ref{Eq:3} is a convolution of  the distribution  of $({X}_b(0),Y_b(t-1))$ with the noise distribution $N$, where the convolution occurs with respect to the  random variable ${X}_b(0)$. Our goal now is to bound  the Fisher information in this convolution with respect to $\tau$.

The Fisher information inequality \cite{Stam1959,Zamir1998}
 bounds the Fisher Information of convolutions of single-variable distributions.
 Essentially, the theorem says that if $x$, $y$ and $\tau$ are real values,  $K(x-\tau)$, $R(x-\tau)$  and $Q(x-\tau)$ are  parameterized families and $K=R\otimes Q$, then $J(K)\leq {1}/({\frac{1}{J(R)}+\frac{1}{J(Q)}})$.
 To apply this inequality to Equation~\ref{Eq:3}, we generalize it to distribution with multiple variables, where only one of them  is  convoluted. We rely on the fact that the random variable  $Y_b(t-1)$ does not depend on  $\tau^*$ (recall, it contains only relative~deviation~measurements). This fact turns out to be sufficient to overcome the potential complication rising from the fact that given the environmental value $\tau^*$, the random variable $X_b(0)$ and the  random multi-variable   ${Y}_b(t-1)$ are no longer independent.
In Appendix \ref{extendingSTAM} we prove Lemma~\ref{lem:STAM}    which extends the Fisher information inequality to our multi-variable (possibly dependent) convolution case, enabling to prove  the inequality 
$\tilde{J}_{b}(t)\leq  {1}/({\frac{1}{J_b(t)} + \frac{1}{J_{N}}}).$
Together with Equation~\ref{Eq:1}, we obtain the required recursive inequality for the FI. This completes the proof~of~the~theorem.
\end{proof}

\bigskip
\section{\large A highly-competitive elementary algorithm}\label{sec:conf}
We define an elementaryelementary  algorithm, termed $\Conf$, and prove that its performances are highly-competitive with those of $\OPT$.
In this algorithm, each sensor $a$ stores in its memory a single parameter $c_a\in\mathbb{R}$ that  represents its \emph{accuracy}  regarding the quality of its current opinion with respect~to~$\tau^*$.
The initial accuracy  of sensor $a$ is  set to $c_a(0)=1/\mbox{var}(\Phi_a)$.
When sensor $a$ observes sensor~$b$ at some round $t$, it receives  $c_b(t)$ and  $\tilde{d}_{ab}(t)$, and acts as follows. Sensor~$a$
first computes the value
$\hat{c}_b(t)={c_b(t)}/({1 + c_b(t) \cdot \var(N)}),$  a reduced accuracy  parameter for sensor $b$ that takes measurement noise into account, and then
proceeds as follows:


\noindent\fbox {{\begin{varwidth}{\dimexpr\textwidth-10\fboxsep-50\fboxrule\relax}
    \parbox{\linewidth}{
\underline{\bf Algorithm $\Conf$}
 \parskip-0.4cm 
\begin{itemize}\parskip0cm
\item
{\bf Update opinion:}  $x_a(t+1)=x_a(t)+\frac{{\tilde{d}_{ab}(t) \cdot \hat{c}_b}(t)}{ {c_a(t)+\hat{c}_b}(t)}. $
\item
{{\bf Update accuracy :}}
$ c_a(t+1)=c_a(t)+\hat{c}_b(t).$
\end{itemize}}
\end{varwidth}
    }
}

\parskip0.1cm

Fix an independent meeting pattern. First, algorithm $\Conf$ is designed such  that at all times, the opinion is  preserved as an unbiased estimator of~$\tau^*$ and the accuracy, $c_a(t)$,  remains equal to the reciprocal of the current variance of the opinion $X_a(t,\Conf)$.
Indeed, the following lemma is proven in Appendix \ref{app:reciprocal}.
\begin{lemma}\label{lem:reciprocal}
 At any round $t$ and for any sensor $a$, we have: (1) the opinion $X_a(t,\Conf)$ serves as an unbiased estimator of  $\tau^*$, and (2)   $c_a(t)=1/\var(X_a(t,\Conf))$.
\end{lemma}

\parskip0cm

We are now ready to analyze the competitiveness of algorithm $\CONF$,
by relating the variance of a sensor~$a$ at round $t$ to  the corresponding {\em FI}, namely, $J_a(t)$.
Recall that
Lemma \ref{for:fish} gives a lower bound on the variance of algorithm $\OPT$ at a sensor $a$, which depends on the corresponding  {\em FI} at the sensor. Specifically, we have: $\var({X}_{a}(t,\OPT))\geq {1}/{J_a(t)}.$
Initially, the {\em FI} $J_a(0)$ at a sensor $a$ equals the Fisher information in the parameterized family $\Phi_a(x-\tau)$ with respect to $\tau$, and hence is at most the initial accuracy  $c_a(0)$ times $\Delta_0$.  In Equation \ref{eq:gaim} (see Appendix \ref{app:lem:c_a(t)}) we show that the gain in accuracy  following an interaction is always at least as large the corresponding upper bound on the gain in Fisher information as given in Theorem \ref{lem:fishernoise}, divided by the initial Fisher-tightness . That is: $c_a(t+1)-c_a(t)\geq \left({{1}/({\frac{1}{J_b(t)}+\frac{1}{J_{N}}}})\right)/{\Delta_0}.$ Informally, this property of $\Conf$ can be interpreted as maximizing the Fisher information flow in each interaction up to an approximation factor of $\Delta_0$. By induction (see proof in  Appendix \ref{app:lem:c_a(t)}), we obtain the following.

\begin{lemma}\label{lem:c_a(t)}
At every round $t$, we have
$
c_a(t)\geq {J_a(t)}/{\Delta_0}.
$
\end{lemma}


\noindent Lemmas  \ref{for:fish}, \ref{lem:reciprocal} and \ref{lem:c_a(t)} can now be combined to yield:
$
\var(X_a(t,\CONF))\leq \Delta_0 \cdot \var(X_a(t,\OPT)).
$ This establishes Theorem \ref{thm:main}. \qed

\bigskip

Note that if $|F|=O(1)$ (i.e., $F$ contains a constant number of distributions, independent of the number of sensors) then  initial Fisher-tightness $\Delta_0$ is a constant, and hence Theorem \ref{thm:main} states that  $\Conf$ is constant-competitive at any sensor and at any time. We now aim to identify those cases where $\Conf$ performs even better. One such case is  when the distributions in $\cF$ as well as the noise distribution $N(\eta)$ are all Gaussians. In this case $\Delta_0=1$ and Theorem~\ref{thm:main} therefore states that the variance of $\Conf$ equals  that of~$\OPT$, for any sensor at at any time. Another case is when $|F|$ is a constant, the noise is Gaussian, but both the population size $n$ and  the round $t$ go to infinity.
In this case, analyzed in Appendix \ref{app:large-times}, as time increases, the performances of $\CONF$ become arbitrarily~close~to~those~of~$\OPT$.
\parskip0.1cm


\bigskip
\section{The Fisher Channel Capacity and convergence times}\label{section:flows}
For a fixed independent meeting pattern, the FI $J_a(t)$ at a sensor $a$ and round $t$ was defined in Section \ref{sub:relative} with respect to algorithm $\OPT$. We note that this definition applies  to any algorithm $A$ as long as it is sufficiently smooth so that the corresponding Fisher informations are well-defined.
This quantity $J_a(t,A)$ would respect the same recursive inequality as state in Theorem \ref{lem:fishernoise}, that is, we have:
$J_a(t+1,A)\leq J_a(t,A) + \frac{1}{\frac{1}{J_b(t,A)} + \frac{1}{J_{N}}}.$
This directly implies the following:
\begin{equation}\label{cor:Fisher-capacity}
J_a(t+1,A)- J_a(t,A)  \leq J_{N}~.
\end{equation}
The inequality above sets a bound of $J_{N}$ for the increase in {\em FI} per round.
In analogy to Channel Capacity as defined by Shannon \cite{Cover} we term this upper bound  as the {\em Fisher Channel Capacity}.

The restriction on information flow as given by the Fisher Channel Capacity can be translated into  lower~bounds for convergence time of algorithm $\OPT$ (and hence also apply for any algorithm).
%
Recall, $\rho$ is the first round when we have more than half of the population satisfying
 $\var(X_a(t))<\epsilon^2$. By Lemma \ref{for:fish}, a sensor, $a$, with variance  smaller than $\epsilon^2$ must have a large {\em FI}, specifically,   $J_a(\rho)\geq 1/\epsilon^2.$ 
To get some intuition on the convergence time,
assume  that the number of sensors is odd, and
 let $J_0$ denote the median initial {\em FI} of sensors (this is  the
median of the {\em FI}, $J_{\Phi_a}$, over all sensors $a$), and assume $J_0\ll 1/\epsilon^2$.
By definition, more than a  half of~the population have initial Fisher information at most $J_0$.
By the Pigeon-hole principle, at least one sensor  has an {\em FI} of, at most,~$J_0$  at $t=0$ and, at least, $1/\epsilon^2$
 at $t=\rho.$
Theorem \ref{eq:time} follows by the fact that, by Equation \ref{cor:Fisher-capacity}, this sensor could increase its {\em FI} by, at most, $J_{N}$ in each observation. 

\bigskip
\section{Conclusion}\label{sec:future}
We provide a fresh approach to the study of clock synchronization, following a purely theoretic distributed algorithmic type of study and employing techniques from information theory. 
 We have focused on arbitrary, yet independent, meeting patterns, and demanded the performances of each sensor to be as high as possible at any point in the execution. We have established lower bounds on the performances of algorithm $\OPT$, the best possible    clock synchronization algorithm operating  under the most liberal version of our model.
We have identified algorithm~$\Conf$, an extremely simple  algorithm whose performances are highly-competitive with those of $\OPT$. Moreover, under Gaussian conditions, the accuracies of sensors under $\Conf$ precisely equal those of~$\OPT$.

Algorithm $\Conf$ is based on storing and communicating a single {\em accuracy} parameter that complements noisy deviation measurements, and on internal computations and update rules that are based on {\em weighted-average} operations.
Our proofs rely on an extensive use of the concept of Fisher information. We use the Cram\'er-Rao bound and  our definition of a \emph{Fisher Channel Capacity} to quantify information flows and to obtain lower bounds on best possible performance. This opens the door for further rigorous quantifications of information flows within collaborative sensors.

Our information theoretic approach allowed us to tackle the clock synchronization problem in dynamic networks. 
In this initial work, we focus on independent meeting patterns which can be considered as representing short times scales in highly dynamic scenarios. As evident by this paper, studying independent meeting patterns is already rather complex. Hence, we leave the study of dependant patterns to future work. Our hope is that studying such extreme dynamic cases will help to provide tools and insights for future work dealing with other dynamic scenarios. 

This work is further relevant to the problem of collective approximation of environmental values by biological groups \cite{Biology}.

\bigskip

\clearpage

\pagenumbering{roman}
\appendix

\renewcommand{\theequation}{A-\arabic{equation}}
\setcounter{equation}{0}
\begin{center}
\textbf{\large{APPENDIX}}
\end{center}

\section{Extending the Fisher inequality}\label{extendingSTAM}
The Fisher information inequality \cite{Stam1959} (see also \cite{Blachman,Rioul,Zamir1998}) applies for three one-variable distribution families  $r(z)$, $p_1(x_1)$, and $ p_2(x_2)$  parameterized by $\mu$ such that  $r$ is a convolution of $p_1$ and $p_2$, that is, $r(z) =  \int p_1(z-t) \cdot p_2(t) dt$. The theorem gives an upper bound of the Fisher information $J^{\mu}_r$ of the family $r(z- \mu)$ (with respect to $\mu$) based on the Fisher information $J^{\mu}_{p_1}$ and $J^{\mu}_{p_2}$ of the families $p_1(x_1- \mu)$, and $ p_2(x_2- \mu)$, respectively. Specifically, the theorem states that:
$(\alpha_1+\alpha_2)^2 J^{\mu}_{r} \leq \alpha_1^2{J^{\mu}_{p_1}} + \alpha_2^2{J^{\mu}_{p_2}}$, for any two real numbers $\alpha_1$ and $\alpha_2$. This in particular implies that  ${1}/{J^{\mu}_{r}} \geq {1}/{J^{\mu}_{p_1}} + {1}/{J^{\mu}_{p_2}}$.

The following lemma extends the Fisher information inequality to the case  where the distributions $p_1$ and $r$ are composed of multiple, not necessarily independent, variables, where the convolution with $p_2$ takes place over one of the variables of $p_1$.
\begin{lemma}\label{lem:STAM}
Let 
$\{p_1(x_1-\tau,\bar{x}_3)\}_{\tau\in\mathbb{R}}$ and ~$\{p_2(x_2-\tau)\}_{\tau\in\mathbb{R}}$ be two $pdf$ families with a translation parameter~$\tau$
such that  $x_1$ and $x_2$ are real variables,  $\bar{x}_3$ is a vector of multiple real valued variables
 and  $J^{\tau}_{p_1(x_1-\tau,\bar{x}_3)}$ and $J^{\tau}_{p_2(x_2- \tau)}$ are the corresponding Fisher information with respect to $\tau$. Let $r(z-\tau, \bar{x}_3)=  \int p_1(t-\tau,\bar{x}_3) \cdot p_2(z-t) dt$ be the convolution of $p_1$ and $p_2$. Then the Fisher information in the family
%
 $\{r(z-\tau, \bar{x}_3)\}_{\tau\in\mathbb{R}}$ with respect to $\tau$
satisfies:
$$
\frac{1}{J^{\tau}_{r(z-\tau,\bar{x}_3)}} \geq \frac{1}{J^{\tau}_{p_1(x_1-\tau,\bar{x}_3)}} + \frac{1}{J^{\tau}_{p_2(x_2- \tau)}}~.
$$
\end{lemma}
\medskip

\begin{proof}
We start by using the definition of $r$ as a convolution over $p_2$ and the first variable of $p_1$:
$$
r(z-\tau,\bar{x}_3)=\int p_1(t-\tau,\bar{x}_3) \cdot p_2(z-t)dt.
$$
 We can insert the density function  $p(\bar{x}_3)$ to rewrite the right hand side  as:
\begin{align*}
 \int p_1(t-\tau|\bar{x}_3)\cdot p(\bar{x}_3) \cdot p_2(z-t)dt \\
=p(\bar{x}_3) \int p_1(t-\tau|\bar{x}_3) \cdot p_2(z-t)dt. \\
\end{align*}
Implying that:
$$r(z-\tau|\bar{x}_3) = \int p_1(t-\tau|\bar{x}_3) \cdot p_2(z-t)dt. $$
We now define the distributions $R(z)=r(z-\tau|\bar{x}_3)$ and $P_1(t)=p_1(t-\tau|\bar{x}_3)$ so that the previous equation becomes:
$$R(z) = \int P_1(t) \cdot p_2(z-t)dt,$$
 for which we  apply the original Lemma  as first proved by Stam \cite{Stam1959} to deduce that for any two real numbers $\alpha_1$ and $\alpha_2$, we have:
$$
(\alpha_1+\alpha_2)^2 J^{\mu}_{R(z-\mu)} \leq \alpha_1^2 \cdot J^{\mu}_{P_1(x_1-\mu)}  +\alpha_2^2 \cdot J^{\mu}_{p_2(x_2-\mu)}.
$$
Note that $J^{\mu}_{P_1(x_1-\mu)}$ is well defined since, for a given $\bar{x}_3$, $P_1(x_1-\mu)$ is proportional to $p_1(x_1-\mu,\bar{x}_3)$  and the Fisher information integral of $p_1(x_1-\mu,\bar{x}_3)$ converges when integrating over all possible values of $\bar{x}_3$. This implies (see \cite{Rioul}) that the Fisher information in the convolution $R(z-\mu)$ is well defined and the equation above holds.
We now multiply both sides of the equation by $p(\bar{x}_3)$ and integrate over $\bar{x}_3$, to obtain:
\begin{equation}\label{eq:4}
(\alpha_1+\alpha_2)^2 \int J^{\mu}_{R(z-\mu)}~ p(\bar{x}_3)~ d\bar{x}_3 \leq \alpha_1^2 \int J^{\mu}_{P_1(x_1-\mu)}~ p(\bar{x}_3)~ d\bar{x}_3   +\alpha_2^2 \int J^{\mu}_{p_2(x_2)} ~ p(\bar{x}_3)~ d\bar{x}_3.
\end{equation}
Plugging in the definitions for  Fisher information and $R(z)$, the integral on the left hand side becomes:

\begin{align*}
\int J^{\mu}_{R(z-\mu)}~ p(\bar{x}_3) d\bar{x}_3 &=\int J^{\mu}_{r(z-\mu-\tau|\bar{x}_3)} p(\bar{x}_3)d\bar{x}_3 \\
 &= \int\int\frac{1}{r(z-\mu- \tau | \bar{x}_3)}\left(\frac{dr(z-\mu- \tau, | \bar{x}_3)}{d\mu}\right)^2~dz ~ p(\bar{x}_3) ~ d\bar{x}_3 \\
&=  \int \int \frac{1}{r(z-\mu-\tau|  \bar{x}_3) p(\bar{x}_3) }\left(\frac{d[r(z-\mu-\tau|  \bar{x}_3)p(\bar{x}_3)]}{d\mu}\right)^2 ~ dz  ~d\bar{x}_3 \\
&=  \int \int \frac{1}{r(z-\mu-\tau,\bar{x}_3 ) }\left(\frac{dr(z-\mu-\tau,\bar{x}_3 )}{d\mu}\right)^2 ~ dz ~ d\bar{x}_3 \\
&=  \int \int \frac{1}{r(z-\mu-\tau,\bar{x}_3 ) }\left(\frac{dr(z-\mu-\tau,\bar{x}_3 )}{d\tau}\right)^2 ~ dz ~ d\bar{x}_3 \\
&=  \int \int \frac{1}{r(\tilde{z}-\tau,\bar{x}_3 ) }\left(\frac{dr(\tilde{z}-\tau,\bar{x}_3 )}{d\tau}\right)^2 ~ d\tilde{z} ~ d\bar{x}_3 \\
&= J^{\tau}_{r(z-\tau,\bar{x}_3)}, \\
\end{align*}
where we used $\tilde{z}=z-\mu$ and the fact that $\bar{x}_3$ is independent of $\tau$.

Similarly, the integral over the first term on the right hand side of Equation \ref{eq:4}  gives $J^{\tau}_{p_1(x_1-\tau,\bar{x}_3)}$.
The last term is:
$$
\int J^{\mu}_{p_2(x_2-\mu)} p(\bar{x}_3)d\bar{x}_3 = J^{\mu}_{p_2(x_2-\mu)} \int  p(\bar{x}_3)d\bar{x}_3 = J^{\mu}_{p_2(x_2-\mu)}=  J^{\tau}_{p_2(x_2-\tau)},
$$
by normalization of the distribution $\bar{x}_3$.

Finally, Equation \ref{eq:4} translates to:
$$
(\alpha_1+\alpha_2)^2 J^{\tau}_{r(z-\tau,\bar{x}_3)} \leq \alpha_1^2 \cdot J^{\tau}_{p_1(x_1-\tau,\bar{x}_3)}  +\alpha_2^2 \cdot J^{\tau}_{p_2(x_2-\tau)},
$$
for any real $\alpha_1$ and $\alpha_2$. Setting $\alpha_1=J^{\tau}_{p_2(x_2-\tau)}$ and $\alpha_2=J^{\tau}_{p_1(x_1-\tau,\bar{x}_3)}$, we finally obtain:
$$
\frac{1}{J^{\tau}_{r(z-\tau,\bar{x}_3)}} \geq \frac{1}{J^{\tau}_{p_1(x_1-\tau,\bar{x}_3)}} + \frac{1}{J^{\tau}_{p_2(x_2-\tau)}},
$$
as desired.
\end{proof}
\bigskip

\section{Proof of Lemma \ref{lem:reciprocal}} \label{app:reciprocal}
 Fix an independent meeting pattern $\cP$ and a $pdf$ assignment, $\Phi_a\in \cF$, for each sensor $a$. Let us now prove the first part of the lemma, namely, that the opinion $X_a(t)$ at any sensor $a$ and round $t$ serves as an unbiased estimator for $\tau^*$.
The claim holds at time zero, and assume by induction that it holds at round $t$. Now consider the case that sensor $a$ observes another sensor $b$ at round $t$. The opinion of $u$ after the interaction, becomes:
\begin{equation}\label{EQ:bias}
x_a(t+1)=x_a(t)+\frac{{\tilde{d}_{ab}(t) \cdot \hat{c}_b}(t)} {{c_a(t)+\hat{c}_b(t)}}=\frac{x_a(t) c_a(t)+x_b(t) \hat{c}_b(t)}{c_a(t)+\hat{c}_b(t)}+ \frac{\eta \cdot \hat{c}_b(t)}{c_a(t)+\hat{c}_b(t)} .
\end{equation}
By induction, $X_a(t)$ and $X_b(t)$ are both unbiased estimators of $\tau^*$.
Recall now that the noise distribution~$N(\eta)$ is centered around zero.
Moreover, observe that at round $t$, the accuracy  at each sensor~$a$, namely $c_a(t)$, is deterministically defined (given the fixed  pattern of meetings, and the assignment of $pdf$'s to the sensors). In particular, at round $t$, the accuracy s $c_a(t)$, $c_b(t)$ as well as $\hat{c}_b(t)$ are all fixed constants.
Equation \ref{EQ:bias} therefore implies the following.
\begin{claim}\label{claim:unbiased}
At any round $t$ and for any sensor $a$, the opinion $x_a(t)$ serves as an unbiased estimator of  $\tau^*$.
\end{claim}
Claim \ref{claim:unbiased} established the first part of the lemma. Let us now turn to prove the second part. This part of the lemma holds for time $t=0$ by definition of $c_a(0)$. Assume by induction that for any sensor $a$ at round $t$ it holds that $c_a(t) =1/\var(X_a(t))$ and consider time $t+1$.
We now consider an interaction between two sensors at round $t$, in which sensor $a$ observes sensor~$b$. The variance of the new opinion of $a$ is:
\begin{align*}
\var(X_a(t+1))&=\var\left(x_a(t)+ \frac{\tilde{d}_{ab} \hat{c}_b(t)}{c_a(t)+\hat{c}_b(t)}\right) = \var\left(\frac{x_a(t) c_a(t)+ \hat{c}_b(t)(x_b(t)+\eta)}{c_a(t)+\hat{c}_b(t)}\right) \\
&= \frac{c_a^2(t) \cdot \var(X_a(t))+\hat{c}_b^2(t) \cdot \var(X_b(t)+\eta)}{(c_a(t)+\hat{c}_b(t))^2} \\
&= \frac{c_a^2(t) \cdot \var(X_a(t))+\hat{c}_b^2(t) \cdot (\var(X_b(t))+\var(N(\eta)))}{(c_a(t)+\hat{c}_b(t))^2} \\
&= \frac{c_a^2(t) \cdot 1/c_a(t)+\hat{c}_b^2(t) \cdot (1/c_b(t)+\var(N(\eta)))}{(c_a(t)+\hat{c}_b(t))^2} = \frac{c_a^2(t) \cdot 1/c_a(t) +\hat{c}_b^2(t) \cdot 1/\hat{c}_b(t)}{(c_a(t)+\hat{c}_b(t))^2} \\
&= \frac{c_a(t)+\hat{c}_b(t)}{(c_a(t)+\hat{c}_b(t))^2} = \frac{1}{(c_a(t)+\hat{c}_b(t))} = \frac{1}{c_a(t+1)}~, \\
\end{align*}
which proves the induction step. This complete the proof of Lemma \ref{lem:reciprocal}.

\bigskip
\section{Proof of Lemma \ref{lem:c_a(t)}}\label{app:lem:c_a(t)}
By Lemma \ref{lem:reciprocal} and the definition of $\Delta_0$, we have $c_a(0)\geq {J_a(0)}/{\Delta_0}$, and hence
Lemma \ref{lem:c_a(t)} holds at time $0$.

Assume by induction that  the lemma holds at round $t$ and consider an interaction at round $t$ when sensor $a$ observes sensor $b$. Let $c_N=1/\var(N(\eta))$. By definition of algorithm $\Conf$, we have:
$$c_a(t+1)-c_a(t) = \hat{c}_b(t)=\frac{1}{1/c_b(t)+1/c_N}.$$
By the induction hypothesis applied on sensor $b$, we have:
$$ \frac{1}{1/c_b(t)+1/c_N}\geq \frac{1}{\frac{\Delta_0}{J_b(t)}+1/c_N}=\frac{1}{\Delta_0}\cdot \frac{1}{\frac{1}{J_b(t)}+\frac{1}{\Delta_0\cdot c_N}}. $$
Again, by definition of $\Delta_0$, we have  $\Delta_0\geq J_{N}/c_N$. Hence:
\begin{equation}\label{eq:gaim}
c_a(t+1)-c_a(t) \geq  \frac{1}{\Delta_0}\cdot \frac{1}{\frac{1}{J_b(t)}+\frac{1}{J_{N}}}.
\end{equation}
This means that the gain in accuracy  at sensor $a$ following an observation of sensor $b$ is up to a multiplicative factor of $\Delta_0$ at least as large the corresponding gain in FI of the sensor (operating under $\OPT$).

Finally, applying the induction hypothesis for sensor $a$ at round $t$, we have $c_a(t)\geq {J_a(t)}/{\Delta_0}$. Plugging this in Equation \ref{eq:gaim}, we obtain:

$$
 c_a(t+1)\geq   \frac{1}{\Delta_0}\cdot \left(J_a(t)+\frac{1}{\frac{1}{J_b(t)}+\frac{1}{J_{N}}}\right)\geq  J_a(t+1)/\Delta_0,
$$
where the second inequality holds by Theorem \ref{lem:fishernoise}. This
completes  the proof of the lemma.
\bigskip

\section{On the performances  of $\CONF$ at large times}\label{app:large-times}
We now investigate the performances of algorithm $\Conf$ at large times,  and show that as time increases, the variance of $\CONF$ becomes arbitrarily close to zero, and moreover, the  performances of $\CONF$ become closer and closer to those of $\OPT$.

The depth $D(\cP)$ of a given independent meeting pattern $\cP$  is defined as the largest round $t$ for which some sensor observes another sensor. For simplicity, we assume  synchronous meeting patterns in which at each round each sensor observes another sensor, but our results can be easily extended to the case where the  number of total interactions per sensor are all roughly the depth~$D(\cP)$. Note that for any  population with $n$ sensors, the depth of an  independent synchronous meeting pattern is at most $\log_2 n$. In particular, the depth is finite for populations of a fixed size. Since our goal is to investigate the behavior of $\CONF$ at large times, whenever we consider a round $t$, we only inspect populations and corresponding meeting patterns for which the depth is at least $t$. 

In the remaining of this section we fix a family of distributions $\cF$ and a  noise distribution $N(\eta)$. 
Given a round $t$, let $\var_{\sup}(t)$ denote the supremum of $\var(X_a(t,\Conf))$, taken over (1) all possible populations $A_n=\{a_1,a_2,\cdots, a_n\}$, for $n=1,2\cdots$, ~(2) all assignments of distributions  $\Phi_a\in \cF$ to the sensors in $A_n$,  (3) all meeting patterns (with depth at least $t$), and (4) all sensors $a_i\in A_n$. Our next claim implies that as time increases, the variance of $\CONF$ becomes arbitrarily close to zero.



\begin{claim}\label{lem:limit-var}
$\lim_{t \rightarrow \infty}  \var_{\sup}(t)=0$.
\end{claim}
\begin{proof}
For a round $t$, let $C_{\inf}(t)=1/\var_{\sup}(t)$. Note that $\var_{\sup}(0)$ is precisely the maximal variance over the distributions in $\cF$. Hence,  $C_{\inf}(0)$ is some positive constant (that depends on $\cF$ only).

 By the definition of $C_{\inf}(t)$ and by Lemma \ref{lem:reciprocal}, it follows that $C_{\inf}(t)$ is the infimum of $c_a(t)$, the accuracy   of a sensor $a$ at round $t$ operating under algorithm $\Conf$, taken over all possible populations, all assignments of distributions  $\Phi_a\in \cF$ to sensors, all meeting patterns, and all sensors $a$.

When sensor $a$ observes sensor $b$ at round $t$, the gain in accuracy  for sensor $a$ is:
$$
\frac{1}{1/c_b(t)+v_N}\geq \frac{1}{2}\cdot\min\{c_b(t),\frac{1}{\var(N(\eta))}\}.
$$
It follows that  $C_{\inf}(t)$ increases in a single round by either at least a multiplicative factor of $3/2$ or by at least an additive constant factor of $1/2\var(N(\eta))$.  This implies that $\lim_{t \rightarrow \infty}  C_{\inf}(t)=\infty$. The proof of the claim now follows by the definition of $C_{\inf}(t)$.
\end{proof}

Since algorithm $\OPT$ is superior over algorithm $\CONF$, the same limit property of the variance applies to algorithm $\OPT$
as well. We now claim that, in fact, if the noise $N(\eta)$ is Gaussian, then the  variances in $\CONF$ and $\OPT$ go to zero at roughly the same speed.

Given a round $t$, let $\kappa(t)$ denote the supremum  of the fraction ${\var(X_a(t,\CONF))}/{\var(X_a(t,\OPT))}$, taken over all possible populations $\{A_n\}_{n=1}^\infty$, all assignments of distributions  $\Phi_a\in \cF$ to sensors $a$ in $A_n$, all meeting patterns, and all sensors $a\in A_n$. 
Note that Theorem \ref{thm:main} implies that for any $t$, we have $\kappa(t)\leq \Delta_0$.

\begin{lemma}\label{conf-limit}
If the noise $N(\eta)$ is Gaussian then $\lim_{t \rightarrow \infty} \kappa(t)=1$.
\end{lemma}
\begin{proof}
Since the noise is Gaussian we have $\var(N(\eta))=1/J_{N}$. Recall the definition of $C_{\inf}(t)$ from the proof of Claim \ref{lem:limit-var}. Note now that as $C_{\inf}(t)$ becomes larger and larger the gain in accuracy  under algorithm $\Conf$ becomes very close to $J_{N}$. Indeed, when sensor $a$ observes sensor $b$ at round $t$, we have: $$c_a(t+1)-c_a(t) =\frac{1}{1/c_b(t)+1/J_{N}}.$$
Specifically, consider now the case that  $c_b(t)>x\cdot J_{N}$, for some large $x$. Here, the increase in accuracy  at $a$ is some quantity $\Delta J(t,\conf)$, satisfying $$\frac{1}{1+1/x} J_{N}\leq \Delta J(t,\conf)\leq J_{N}.$$
The Cram\'er-Rao bound and Lemma \ref{lem:reciprocal} imply that $J_b(t,\OPT)\geq c_b(t)$, and hence,
 $J_b(t,\OPT)>x\cdot J_{N}$. This, together with Theorem \ref{lem:fishernoise},  implies that at round $t$, the increase $\Delta J(t,\opt)$ in Fisher information of $a$ under algorithm $\OPT$  is
some quantity satisfying $$\frac{1}{1+1/x} J_{N}\leq \Delta J(t,\opt) \leq J_{N}.$$
 Hence $\Delta J(t,\conf)$ and $\Delta J(t,\opt)$ are the same quantity up to a multiplicative factor of $\frac{1}{1+1/x}$. Finally, 
 since $\lim_{t \rightarrow \infty}  C_{\inf}(t)=\infty$ (see the proof of Claim \ref{lem:limit-var}), it follows that $x$ goes to infinity as $t$ goes to infinity. We thus get $\lim_{t \rightarrow \infty} \kappa(t)=1$, which establishes the proof of the lemma.
\end{proof}


{\small 
\begin{thebibliography}{10}

%

%
\bibitem{Ang}
D. Angluin. 
\newblock {\em  Local and Global Properties in Networks of Processors}
\newblock  STOC 82-93, 1980.

\bibitem{Attiya}
H. Attiya, A. Herzberg, and S. Rajsbaum. 
\newblock {\em  Optimal Clock Synchronization under Different Delay Assumptions.} 
\newblock SIAM J. Comput. 25(2), 369-389, 1996.

\bibitem{info-comm1}
Z. Bar-Yossef, T. S. Jayram, R. Kumar, D. Sivakumar.
\newblock {\em Information Theory Methods in Communication Complexity.}
\newblock IEEE Conference on Computational Complexity, 93-102, 2002.


\bibitem{Welch}
S. Biaz, and J. L. Welch.
\newblock {\em Closed form bounds for clock synchronization under simple uncertainty assumptions.}
\newblock  Inf. Process. Lett. 80(3), 151-157, 2001.



\bibitem{Blachman}
N. M. Blachman.
\newblock {\em The convolution inequality for entropy powers.}
\newblock  IEEE Transactions on Information Theory 11(2), 267-271, 1965.

%
\bibitem{Chaudhari}
Q. Chaudhari, E. Serpedin, and Y.C. Wu, 
\newblock {\em  Improved estimation of clock
offset in sensor networks.}
\newblock  Proc. IEEE Int. Conf. Communications (ICC), 2009.

\bibitem{Chaudhari2}
Q. Chaudhari, E. Serpedin, and K. Qaraqe. 
\newblock {\em On minimum variance
unbiased estimation of clock offset in distributed networks}. 
\newblock IEEE Trans.
Inform. Theory, vol. 56, no. 6,  2010.


\bibitem{Cover}
T. M. Cover  and J.A. Thomas.
\newblock {\em Elements of Information Theory}.
\newblock  {\em John Wiley \& Sons}, 2nd edition, 2006.

\bibitem{Cristian}
F. Cristian.
\newblock {\em  Probabilistic Clock Synchronization.}
\newblock Distributed Computing 3(3): 146-158, 1989.




\bibitem{Dolev}
D. Dolev, J. Halpern, B. Simons, and R. Strong. 
\newblock {\em Dynamic fault-tolerant clock synchronization.}
\newblock Journal of the ACM, 42(1), 143--185, 1995.


\bibitem{Elson}
J. Elson, L. Girod, and D. Estrin. Fine-Grained
\newblock {\em Network Time Synchronization Using Reference Broadcasts.}
\newblock  ACM SIGOPS Operating Systems Review, 36, 147--163, 2002.

 \bibitem{Elson2}
 J. Elson and K. Romer. 
\newblock {\em Wireless Sensor Networks: A New Regime for Time Synchronization.}
\newblock  Proc. First
Workshop on Hot Topics In Networks (HotNets-I), Princeton, New Jersey. Oct. 2002.

 \bibitem{Breathe}
Of. Feinerman, B. Haeupler, and A. Korman. 
\newblock {\em Breathe before speaking: efficient information dissemination despite noisy, limited and anonymous communication.}
\newblock  PODC, 114-123, 2014.



\bibitem{Fetzer}
C. Fetzer and F. Cristian.
 \newblock {\em Integrating External and Internal Clock Synchronization.}
  \newblock  Real-Time Systems 12(2): 123-171, 1997.







\bibitem{NetworkIT}
A. El Gamal and Y. Kim.
\newblock {\em Network Information Theory.}
\newblock Cambridge University Press, 709pp, 2012.






\bibitem{Gubner}
J. Gubner.
\newblock {\em Distributed Estimation and Quantization}
\newblock IEEE Transactions on Information Theory, 39(4), 1993.

\bibitem{Tempo}
 R. Gusella and S. Zatti. 
\newblock {\em TEMPO - A Network Time Controller for a Distributed Berkeley UNIX System. }
\newblock IEEE Distributed Processing Technical Committee Newsletter 6, No S12-2, 1984.

 \bibitem{Jeske}
D. Jeske. 
\newblock {\em On the maximum likelihood estimation of clock offset.} 
\newblock IEEE
Trans. Commun., vol. 53, no. 1, 2005.


\bibitem{KM09}
S. Kar and J. M.F. Moura.
\newblock {\em Distributed Consensus Algorithms in Sensor Networks With Imperfect Communication: Link Failures and Channel Noise}
\newblock IEEE Transactions on SIgnal Processing 57(1), 355--369, 2009.

\bibitem{kempe}
D. Kempe, A. Dobra, J. Gehrke.
\newblock {\em Gossip-Based Computation of Aggregate Information.}
\newblock FOCS 2003: 482-49.

\bibitem{KK08}
R. Koetter and F. R. Kschischang.
\newblock {\em Coding for errors and erasures in random network coding}.
\newblock IEEE Transactions on Information Theory 54 no. 8, 3579--3591, 2008.


\bibitem{Biology}
A. Korman, E. Greenwald and O Feinerman.
\newblock {\em  Confidence Sharing: an Economic Strategy for Efficient Information Flows in Animal Groups}.
\newblock PLOS Computational Biology, 10(10), 2014.


\bibitem{kuhn}
F. Kuhn,  C. Lenzen, T. Locher, and R. Oshman.
\newblock {\em Optimal gradient clock synchronization in dynamic networks.}
	\newblock  PODC, 430-439, 2010.



\bibitem{Lamport}
L. Lamport.
\newblock {\em Proving the Correctness of Multiprocess Programs.}
\newblock IEEE Transactions on Software Engineering (2): 125--143. 1977.



\bibitem{Leng}
M. Leng and Y. Wu. 
\newblock  {\em On joint synchronization of clock offset and skew
for wireless sensor networks under exponential delay}.
\newblock   in Proc. IEEE ISCAS. 2010.

\bibitem{Lenzen2}
C. Lenzen, T. Locher, and R. Wattenhofer.
\newblock  {\em Tight Bounds for Clock Synchronization.}
\newblock  J. of the ACM, 57(2),  2010.

\bibitem{Lenzen}
C. Lenzen, P. Sommer, and R. Wattenhofer.
\newblock  {\em PulseSync: An Efficient and Scalable
Clock Synchronization Protocol.}
\newblock  ACM/IEEE Transactions on Networking (TON), March 2014.



\bibitem{Lenzen3}
C. Lenzen, T. Locher, P. Sommer, and R. Wattenhofer.
\newblock  {\em Clock Synchronization: Open Problems in Theory and Practice.}
\newblock  Proc. 36th International Conference on Current Trends in Theory and Practice of Computer Science (SOFSEM), January 2010.

\bibitem{Lynch}
J. Lundelius and N. Lynch. 
\newblock  {\em An Upper and Lower Bound 
for Clock Synchronization.}
\newblock Information and Control  62, 190-204 (1984).


\bibitem{McNamara1987a}
J. M. McNamara and A. I. Houston.
\newblock {\em Memory and the efficient use of information.}
\newblock Journal of theoretical biology 125(4), 385--395, 1987.




%
\bibitem{NPT}
D.L. Mills.
\newblock {\em Internet time synchronization: the network time protocol.}
\newblock  IEEE Transactions of Communications 39 (10)  1482--1493, 1991.

\bibitem{Mills}
D.L. Mills. 
\newblock {\em  Improved algorithms for synchronizing computer network clocks.}
\newblock  IEEE/ACM Trans. Networks 3, 3, 1995.


\bibitem{Boaz1}
R. Ostrovsky and B. Patt-Shamir. 
\newblock {\em Optimal and Efficient Clock Synchronization Under Drifting Clocks.}
\newblock  PODC 1999, 3-12.


\bibitem{Boaz2}
B. Patt-Shamir and S. Rajsbaum. 
\newblock {\em  A theory of clock synchronization.}
\newblock STOC 1994, 810-819.



\bibitem{Peleg:book}
D. ~Peleg.
\newblock {\em Distributed Computing: A Locality-Sensitive Approach}.
\newblock SIAM, 2000.

\bibitem{Rioul}
O. Rioul.
\newblock {\em Information theoretic proofs of entropy power inequalities.}
\newblock  IEEE Transactions on Information Theory, Vol. 57, No. 1, pp. 33-55, 2011.





\bibitem{Sichitiu}
M. Sichitiu and C. Veerarittiphan. 
\newblock {\em Simple, Accurate Time
Synchronization for Wireless Sensor Networks.}
\newblock  In Proc. IEEE Wireless Communications and Networking Conference
(WCNC), 2003.



\bibitem{survey}
F. Sivrikaya and B. Yener. 
\newblock {\em Time synchronization in sensor networks: a survey.}
\newblock  IEEE Network 18(4): 45-50, 2004.




\bibitem{Shannon}
C. Shannon.
\newblock {\em A Mathematical Theory of Communication}.
\newblock Bell System Technical Journal 27(3), 379--423, 1948.

\bibitem{Solis}
R. Solis, V. Borkar, and P. R. Kumar. 
\newblock {\em A New Distributed
Time Synchronization Protocol for Multihop Wireless Networks.}
\newblock In Proc. 45th IEEE Conference on Decision and Control
(CDC),  2006.


\bibitem{Srikanth}
T. K. Srikanth and S. Toueg. 
\newblock {\em Optimal Clock Synchronization.} 
\newblock J. ACM, 34(3), 626--645, 1987.

\bibitem{Stam1959}
A. J. Stam.
\newblock {\em Some inequalities satisfied by the quantities of information of Fisher and Shannon.}
\newblock Inform. and Control, 2, 101-112. 1959.


\bibitem{Xiao}
L. Xiao, S. Boyd	 and S. Lall.
\newblock {\em A scheme for robust distributed sensor fusion based on average consensus}.
\newblock Proc of the 4th international symposium on Information processing in sensor networks (IPSN), 2005.


\bibitem{survey2}
B. Sundararaman, U. Buy, and A. D. Kshemkalyani.
\newblock {\em Clock synchronization for wireless sensor networks: a survey.}
\newblock Ad Hoc Networks 3, 281--323, 2005.



\bibitem{Detection}
R. Viswanathan, and P.K. Varshney.
\newblock {\em  Distributed detection with multiple sensors I. Fundamentals}.
\newblock Proceedings of the IEEE, 1997.




\bibitem{Qasim}
Y.C. Wu, Q. M. Chaudhari, and E. Serpedin. 
\newblock Clock Synchronization of Wireless Sensor Networks. 
\newblock IEEE Signal Process. Mag. 28(1): 124-138, 2011.


\bibitem{Zamir1998}
R. Zamir.
\newblock {\em A proof of the Fisher Information inequality via a data processing arguement.}
\newblock IEEE Trans Inf Theory, 482--491, 2003.



\end{thebibliography}
}
\end{document}